\documentclass[conference]{IEEEtran}
\IEEEoverridecommandlockouts
\usepackage{cite}
\usepackage{float}
\usepackage{amsmath,amssymb,amsfonts}
\usepackage{amsthm}
\usepackage{algorithmic}
\usepackage{graphicx}
\usepackage{subfigure}
\usepackage{textcomp}
\usepackage{xcolor}
\usepackage{array}
\usepackage{makecell}
\usepackage{graphicx}
\usepackage{multirow}
\usepackage{enumitem}
\usepackage{comment}
\usepackage{threeparttable}
\usepackage[hidelinks]{hyperref}
\def\BibTeX{{\rm B\kern-.05em{\sc i\kern-.025em b}\kern-.08em
		T\kern-.1667em\lower.7ex\hbox{E}\kern-.125emX}}

\newcommand{\bfy}{\mathbf{y}}
\newcommand{\bfx}{\mathbf{x}}

\newcommand{\bfu}{\mathbf{u}} 
\newcommand{\bfU}{\mathbf{U}}

\newcommand{\bfB}{\mathbf{B}}
\newcommand{\domega}{\Delta\boldsymbol{\omega}}

\newcommand{\mbf}[1]{\mathbf{#1}}

\newcommand{\what}{\widehat}
\newcommand{\wtilde}{\widetilde}

\DeclareMathOperator*{\argmin}{arg\,min}

\newcommand{\transpose}{\mathsf{T}}

\newcommand{\bfzeros}{\mathbf{0}}
\newcommand{\vertiii}[1]{{\left\vert\kern-0.25ex\left\vert\kern-0.25ex\left\vert #1 
		\right\vert\kern-0.25ex\right\vert\kern-0.25ex\right\vert}}

\newcommand{\bfA}{\mathbf{A}}
\newcommand{\bfG}{\mathbf{G}}

\newcommand{\bfC}{\mathbf{C}}

\newcommand{\norm}[1]{\left\lVert#1\right\rVert}

\newtheorem{theorem}{\bf \emph{Theorem}}

\newtheorem{proposition}[theorem]{Proposition}

\newtheorem{assumption}[theorem]{Assumption}

\usepackage[linesnumbered,ruled,vlined]{algorithm2e}

\begin{document}
	
	\title{Localizing Single and Multiple Oscillatory Sources: A Frequency Divider Approach \\
		\thanks{\color{black} This work is funded in part by National Science Foundation (NSF) grants OAC-1934766, ECCS-2246658, and ECCS-2145063, and the Power Systems Engineering Research Center (PSERC) grant S-100.}
	}
	
	\author{\IEEEauthorblockN{Rajasekhar Anguluri,~\IEEEmembership{Member,~IEEE}, and Anamitra Pal,~\IEEEmembership{Senior Member,~IEEE}}
		\IEEEauthorblockA{\textit{School of Electrical, Computer, and Energy
				Engineering} \\
			\textit{Arizona State University, Tempe, AZ 85281 USA}\\
			\{rangulur, anamitra.pal\}@asu.edu}}
	
	\maketitle
	
	\begin{abstract}
		Localizing sources of troublesome oscillations, particularly forced oscillations (FOs), 
  in power systems has received considerable attention over the last few years. This is driven in part by the massive deployment of phasor measurement units (PMUs) that
  capture these oscillations when they occur;
  and in part by the increasing incidents of FOs due to malfunctioning components, wind power fluctuations, and/or cyclic loads. Capitalizing on the 
  \emph{frequency divider} formula 
  of \cite{milano2016frequency}, we develop methods to localize single and multiple oscillatory sources using bus frequency measurements. The method to localize a single oscillation source does not require knowledge of network 
  parameters. However, the method for localizing FOs caused by multiple sources requires this knowledge.
  We explain the reasoning behind this knowledge difference as well as demonstrate the success of our methods
  for source localization in multiple test systems.
	\end{abstract}
	
	
	\section{Introduction}
	\textcolor{black}{Low frequency oscillations can limit power transfer, cause damage to equipment, and lead to catastrophic blackouts \cite{gupta2020coordinated,gupta2021coordinated}. These oscillations are categorized as natural oscillations (NOs) and forced oscillations (FOs). The former occur due to power variations on long transmission lines, control actions from high-gain exciters, or control devices with poorly tuned settings. They can be suppressed by power system stabilizers (PSSs) or by regulating specific outputs of DC lines, flexible AC transmission system devices, and energy storage systems \cite{pal2013applying}.}
 {\color{black} 
Initiated by periodically exogenous disturbances, FOs can persist and propagate to the entire system if the driving source is not promptly removed \cite{meng2020time}. Common sources of FOs include malfunctioning PSSs, mechanical oscillations in the generator-turbine unit, or a load experiencing cyclic changes \cite{ghorbaniparvar2017survey}.
  Yet another emerging source of FOs is wind turbines. Here, FOs occur due to periodic power fluctuations caused by wind shear, tower shadow effects, or vibrations of floating offshore wind turbines \cite{su2013mitigation}. Thus, detecting and localizing FOs is crucial for the safe, reliable, and resilient operation of the power system.}

 Detection precedes localization and helps in distinguishing between NOs and FOs 
 (captured in the phasor measurement unit (PMU) data); whereas localization helps in mitigating the oscillations either by disconnecting the sources or by dampening the oscillation using counteracting signals. Detecting oscillations is generally easier and has been well-studied 
 \cite{HY-YL-PZ-ZD:17}. 
 {\color{black} Conversely,
 localizing FOs---especially with PMU data and with minimal knowledge of electro-mechanical dynamics---is a difficult problem, and the focus of this paper. 
	
	

In general, one could use voltages, currents, or frequencies of PMUs to localize FOs. However, we focus on the frequency data because 
the \emph{frequency divider} formula---advocated by the authors in \cite{milano2016frequency}---quite elegantly expresses the bus frequency variations as a linear superposition of the rotor speed deviations. 
Using this formula and without the explicit knowledge of electro-mechanical dynamics, this research put forth two methods to localize a single and multiple sources of FOs. 
A quick summary of our main contributions is as follows:
	\begin{enumerate}[itemsep=0.1cm]
			\item \emph{Single source}: 
   Demonstrating that magnitudes of the bus frequency variations at the source bus is greater than that of the non-source buses. 
   This allows to localize the single source from \textit{magnitudes alone} without knowledge of bus topology and line or reactance parameters. 
		    \item \emph{Multiple sources}: Through multiple simulations showcasing the failure of magnitude-based localization in this scenario. To identify sources, the total least squares (TLS) technique is employed, with the underlying model based on the frequency divider formula. Here, knowledge of bus topology and line or reactance parameters is necessary. 
            \item Illustrating the superior performance and limitations of the methods on a suite of benchmark systems: the WECC-9 bus, IEEE-14 bus, and IEEE 39-bus systems. 
	\end{enumerate}

{\emph{Key novelty}: In contrast to the existing methods (discussed below), our methods are source agnostic; they handle arbitrary disturbances from synchronous/converter-based resources. 
For multiple source locations, we account for errors in network parameters, which is new in the context of source localization.}
 
	

	
	
	
	
	\vspace{0.3mm}
	\textit{Related literature}: Because oscillations are easily analyzed in the complex plane, many methods rely on spectral properties of transfer functions between source inputs and outputs (see \cite{NZ-MG-SA:17, ML-SL-DG-DW:20, SR-WJ-NN-BL:21}). In \cite{estevez2022complex,estevez2022forced}, the authors leveraged the concept of dissipating energy between source and sink nodes to localize sources. Because the number of sources and harmonics in the oscillations are sparse, 
    \cite{SCC-PV-KT:18, anguluri2022complex, TH-NMF-PRK-LX:20}
    relied on $\ell_1$-regularized convex optimization algorithms to localize sources. 
    In addition, several data-driven and time-domain based methods were also proposed (see the works that cited \cite{TESTCASE}).
 
	

 Reference \cite{ortega2021source} also relied on the frequency divider formula for source localization; our work differ from the cited work on three fronts. First, the single-source localization method in this paper does not need bus topology. Second, \cite{ortega2021source} relies on the discrete Fourier transform (DFT), requiring the injected signals to have distinct harmonics. In addition, DFT methods suffer from the problem of  spectral leakage and thresholding (see \cite{anguluri2022complex,jafarpisheh2021robust}). Methods in this work are signal agnostic. Third, \cite{ortega2021source} only considers measurement noise whereas we also consider errors in the network parameters.}


	\section{Problem Setup and Preliminaries}\label{sec: prelims}

    Recall the frequency divider formula in \cite{milano2016frequency} that relates the change in rotor speed deviations of synchronous machines to the frequency variations at network buses. Formally, 
    \begin{align}\label{eq: frequency divider}
        \bfB_\text{BG}\domega_G(t)\approx-\bfB_\text{BB}\domega_B(t), \quad t\geq 0, 
    \end{align}
    where $\domega_B(t)\in \mathbb{R}^N$ is the vector of bus frequency variations; $\domega_G(t)\in \mathbb{R}^n$ is the vector of rotor speed deviations of synchronous machines; and the matrices are 
    \begin{itemize}
        \item $\bfB_\text{BB}=\bfB_\text{BUS}+\bfB_\text{GG} \in \mathbb{R}^{N\times N}$, 
        \item $\bfB_\text{BUS}\in \mathbb{R}^{N\times N}$ is the imaginary part of the complex-valued network admittance, 
        \item $\bfB_\text{BG}\in \mathbb{R}^{N\times n}$ encodes susceptances between the generators and the buses to which the generators are connected, 
        \item $\bfB_\text{GG}\in \mathbb{R}^{N\times N}$ is a diagonal matrix with $\bfB_\text{GG}(i,i)\ne 0$ if there is a synchronous machine connected to the $i$-th bus of the network, and its value is the inverse of the machine's internal reactance.
    \end{itemize}
    

The formula in \eqref{eq: frequency divider} is an approximation to the reality because it excludes line and internal reactance errors, and the errors in $\domega_B(t)$. Including these errors in \eqref{eq: frequency divider} gives us  
    \begin{align}\label{eq: frequency divider errors}
    (\bfB_\text{BG}\!+\!\Delta\bfB_\text{BG})\domega_G(t)\!=\!-(\bfB_\text{BB}\!+\!\Delta\bfB_\text{BB})(\domega_B(t)+\mathbf{e}(t)), 
    \end{align}
where $\Delta\bfB_\text{BB}=\Delta\bfB_\text{BUS}+\Delta\bfB_\text{GG} \in \mathbb{R}^{N\times N}$ is the cumulative error. 
$\Delta\bfB_\text{BG}$ and $\Delta\bfB_\text{BB}$ capture the variations in the network parameters, and lie within $\pm30\%$ of their database values \cite{kusic2004measurement}. The errors in the bus frequency measurement variations obtained from PMUs, $\mathbf{e}(t)$, usually lie within $\pm0.008~\mathrm{Hz}$ (for U.S. power grids) \cite{biswas2023open}.
The modeling errors in $\domega_G(t)$ are excluded as they are insignificant compared to the magnitudes of the rotor oscillations. 


\section{Localizing a single source}
This section introduces a practically feasible technique to localize oscillation produced by one generator (which can be either synchronous-based or converter-based).
We also assume no knowledge of network admittance and reactance parameter matrices.
Finally, for this initial study, we consider the errors in parameters and noise to be negligible (this will be relaxed in the next section).
In such a scenario, \eqref{eq: frequency divider} can be expressed with an equality sign.
Let $\mathbf{H}^+$ denote the Moore-Penrose pseudo-inverse of the matrix $\mathbf{H}$. Then, from \eqref{eq: frequency divider}, it follows that 
\begin{align}\label{eq: single-no-noise-model}
    \domega_B(t)=-\bfB_\text{BB}^+\bfB_\text{BG}\domega_G(t).
\end{align}

Several properties of $\bfB_\text{BB}^+\bfB_\text{BG}$ have already been studied in the literature \cite{milano2020frequency, milano2018model}. For example, this product matrix is dense (it has more non-zeros) even when matrices $\bfB_\text{BB}$ and $\bfB_\text{BG}$ are sparse. Our work points out another important property ignored in the earlier studies for transmission systems and only recently explored in distribution systems for fault identification \cite{sajan2020realistic}: \emph{The matrix $\bfB_\text{BB}^+$ is diagonally dominant} \cite{van2017pseudoinverse}. 
Note
that in a 
diagonally dominant matrix, the magnitude of the $i$-th diagonal entry is greater than or equal to the sum of the magnitudes of the entries in the $i$-th row or column.

{\color{black} The matrix $\bfB_\text{BB}^+$ is not only diagonally dominant but every entry is non-negative \cite[Chapter 2]{horn2012matrix}. On the other hand, the $i$-th column of $\bfB_\text{BG}$ is the vector of zeros except at the $i$-th entry, where a generator is connected to the $i$-th network bus. With these facts in place, we state our first result.}

 

\begin{proposition}\label{prop: single source} Let the (synchronous/converter-based) generator
connected to the $i$-th network bus be inducing FOs. Then, $|\domega_{\mathbf{B},i}(t)|\geq |\domega_{\mathbf{B},j}(t)|$, for $i\ne j \in \{1,\ldots, N\}$, where $\domega_{\mathbf{B},i}(t)$ is the $i$-th component of $\domega_\mathbf{B}(t)$.   
\end{proposition}

\begin{proof}
Without loss of generality assume that the $n$ generators
be connected to the buses $1,\ldots,n$ of the network, and set $i=1$. Then, the expression in \eqref{eq: single-no-noise-model} can be expanded as
\begin{align*}
    \begin{bmatrix}
        \domega_{\bfB,1}(t)\\
        \domega_{\bfB,2}(t)\\
        \vdots
        \\
        \domega_{\bfB,N}(t)
    \end{bmatrix}&=\bfB^+\underbrace{\begin{bmatrix}
        b_{11} & 0 & \ldots & 0\\
        0 & b_{22} & \dots  & 0\\
        \vdots & \vdots & & \vdots\\
        0 & 0 & \ldots & 0
    \end{bmatrix}}_{\bfB_{\bfB\bfG}}    \begin{bmatrix}
        \domega_{\mathbf{G},1}(t)\\
        0\\
        \vdots
        \\
        0
    \end{bmatrix}_{n\times 1}, 
\end{align*}
where $b_{ii}>0$, for all $i\in \{1,\ldots,n\}$. Observe that $\bfB_{\bfB\bfG}$ is not a diagonal matrix. It has only $n$ columns with the structure described above. By matrix multiplication, we have
\begin{align*}
    \begin{bmatrix}
        \domega_{\bfB,1}(t)\\
        \domega_{\bfB,2}(t)\\
        \vdots
        \\
        \domega_{\bfB,N}(t)
    \end{bmatrix}&=\begin{bmatrix}
        \bfB^+(1,1)\\
        \bfB^+(2,1)\\
        \vdots
        \\
        \bfB^+(N,1)
    \end{bmatrix}b_{11}\domega_{\bfG,1}(t). 
\end{align*}
Hence,
\begin{align*}
    |\domega_{\bfB,1}(t)|&=|\bfB^+(1,1)||b_{11}\domega_{\bfG,1}(t)|\\
    &\geq |\bfB^+(j,1)||b_{11}\domega_{\bfG,1}(t)|=|\domega_{\bfB,j}(t)|,
\end{align*}
for all $N\geq j\geq 1$. The inequality follows because $\bfB$ is both row- and column-wise diagonally dominant. 
\end{proof}

We explain the utility of this result to localize a single FO source. Let every network bus have a PMU installed and the oscillations causing generator be at the $i$-th bus. Then, the $i$-th PMU's frequency measurements have greater magnitude than those of other PMUs. This result is intuitive because the energy (and hence, the magnitude) of the oscillation dissipates as it traverses across the network. 
Fig.~\ref{fig:figure_label1} illustrates Proposition \ref{prop: single source} on WECC-9 bus system with three generators and IEEE 14-bus system with five generators. 

Fig.~\ref{fig:figure_label1} shows the bus frequency variations at select buses for one source (bus 1). These simulations bear out the predictions discussed above. 
{\color{black} We also infer that localizing a single source based on magnitudes is easier on a larger system such as the IEEE 14-bus system compared to a smaller system such as the WECC 9-bus system.}
This is because the diagonal entry of the nodal admittance matrix $\mathbf{B}_\text{BUS}$ is the negative of the algebraic sum of the off-diagonal entries. So, for a larger system, the diagonal term (of both $\mathbf{B}_\text{BUS}$ and $\mathbf{B}_\text{BB}^+$) has higher magnitude than any of the individual off-diagonal terms.

{\color{black} Fig.~\ref{fig:figure_label2} illustrates the bus frequency variations at select buses for two sources. For the IEEE 14-bus system, the sources are at 1 and 8. Magnitudes of variations at the source buses are not uniformly (across time) higher than those at the non-source buses. This simulation shows that the prediction in Proposition \ref{prop: single source} might not hold when there is more than one source bus. 
This is because each bus frequency variation is given by a weighted sum of the rotor speed variations of the source buses. Here, the weights are the coefficients of $\mathbf{B}^+_\text{BB}$. This weighted sum for the multiple sources case distorts the magnitude-dominant property described in Proposition \ref{prop: single source}. This fact holds regardless of the oscillation frequencies in $\domega_G(t)$.}


\begin{figure}[!t]
  \centering
  \subfigure[Bus frequency variations in WECC 9-bus system. The single oscillatory source is at Bus 1, which has the highest magnitude of frequency variations.]{\includegraphics[width=1.0\linewidth]{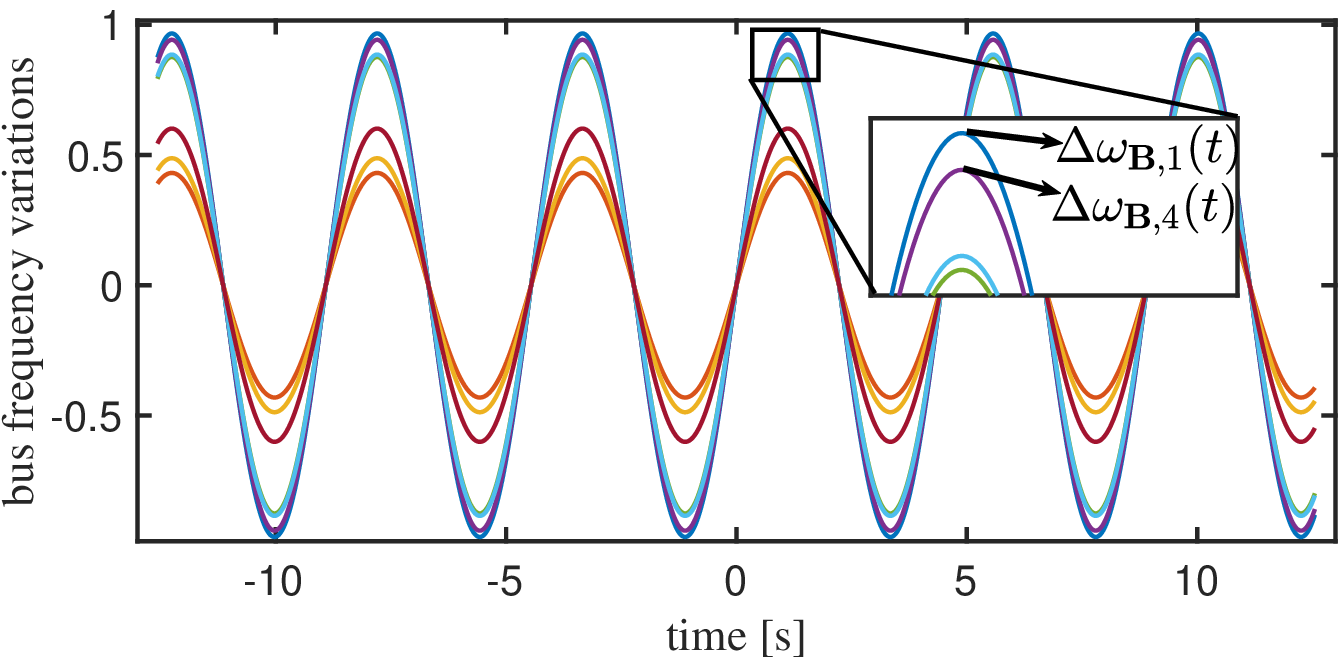}%
  \label{fig:subfigure11}}
  \hfil
  \subfigure[Bus frequency variations in IEEE 14-bus system. The single oscillatory source is at Bus 1, which has the highest magnitude of frequency variations.]{\includegraphics[width=1.0\linewidth]{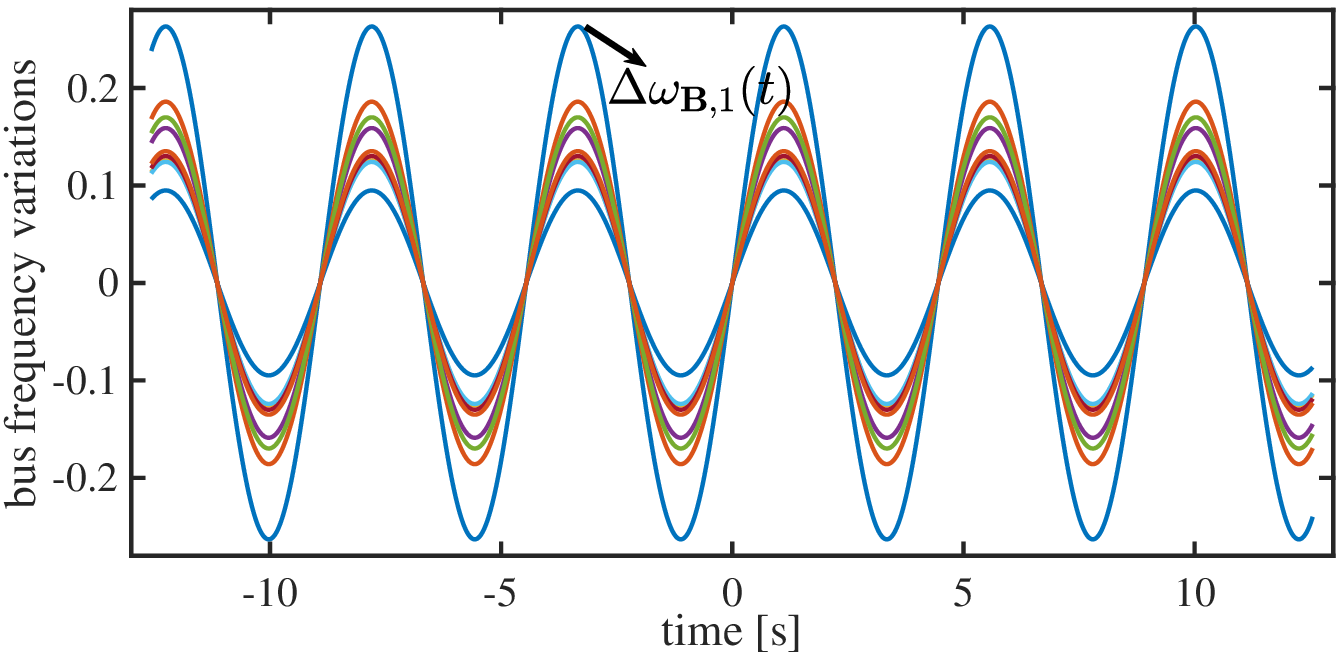}%
  \label{fig:subfigure12}}
  \caption{\small Simulations supporting the result in Proposition \ref{prop: single source}.}
  \label{fig:figure_label1}
\end{figure}

\begin{figure}[!t]
  \centering
  \subfigure[Bus frequency variations in WECC 9-bus system. Magnitudes of the bus frequencies at the sources buses (1 and 3) are not dominant at all times.]{\includegraphics[width=1.0\linewidth]{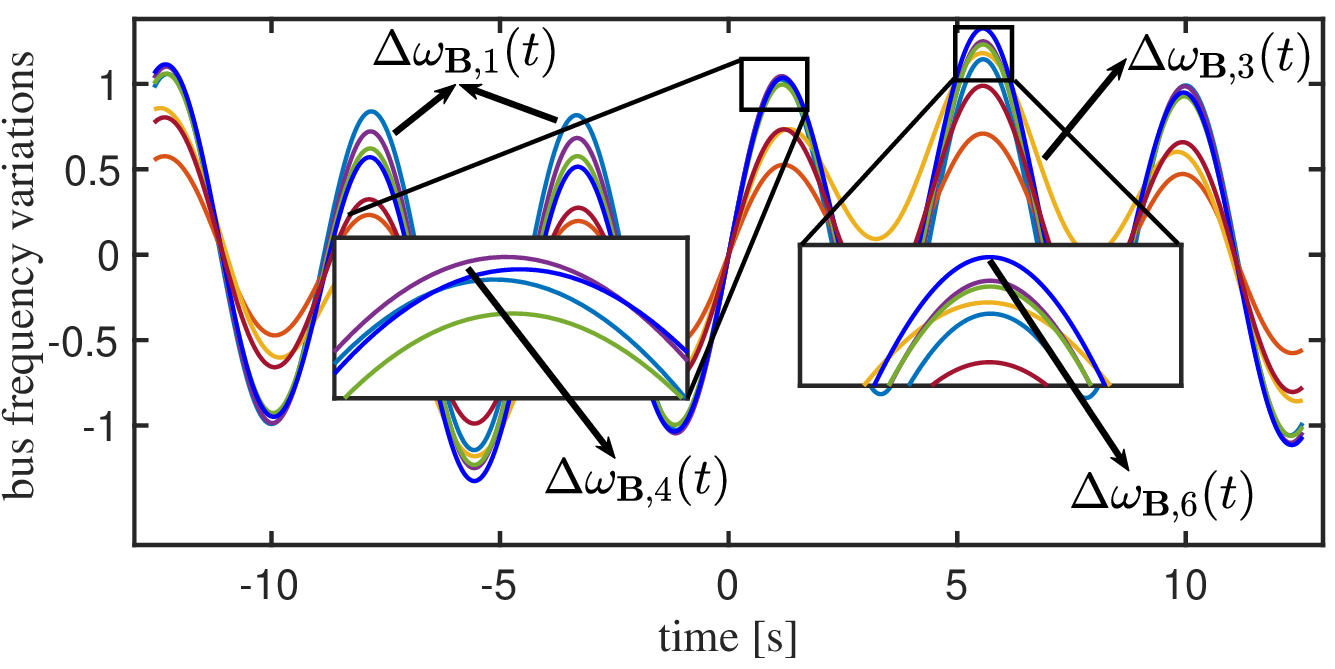}%
  \label{fig:subfigure21}}
  \hfil
  \subfigure[Bus frequency variations in IEEE 14-bus system. Magnitudes of the bus frequencies at source bus 1 is not dominant at all times.]{\includegraphics[width=1.0\linewidth]{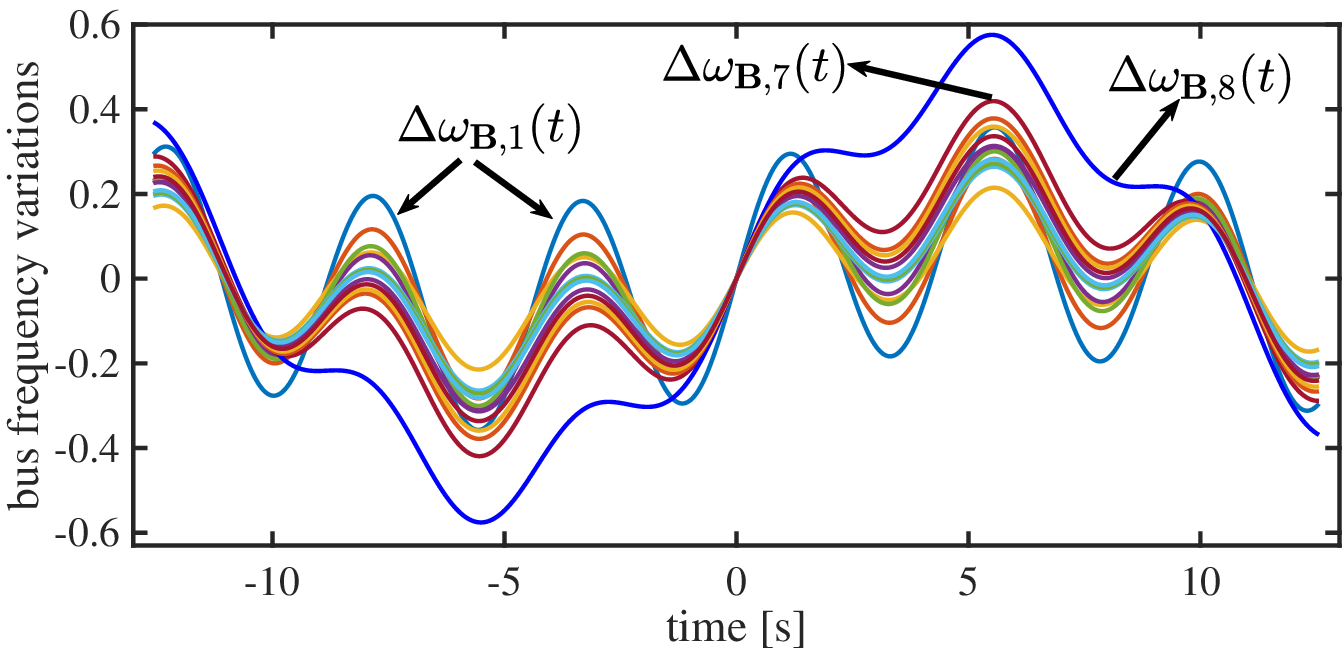}%
  \label{fig:subfigure22}}
 \caption{\small Simulations not supporting the result in Proposition \ref{prop: single source}.}
  \label{fig:figure_label2}
\end{figure}

\section{Localizing multiple sources}\label{sec: multiple sources}
Having discussed the difficulties relying on frequency magnitudes for localizing multiple sources in the previous section, we now discuss the ability of the TLS technique to localize oscillation sources with single or multiple harmonics.
Unlike the setting in the previous section, this setting assumes the knowledge of network admittance matrix and internal reactance parameters. Accounting for errors in these parameters and noise in PMU measurements, rewrite \eqref{eq: frequency divider errors} as 
    \begin{align}\label{eq: frequency divider errors1}
    (\bfB_\text{BB}\!+\!\Delta\bfB_\text{BB})^+(\bfB_\text{BG}\!+\!\Delta\bfB_\text{BG})\domega_G(t)\!=\!-(\domega_B(t)+\mathbf{e}(t)).
    \end{align}
Define the error term $\Delta_e=(\bfB_\text{BB}\!+\!\Delta\bfB_\text{BB})^+-\bfB_\text{BB}^+$. Then, the matrix multiplying $\domega_G(t)$ can be expanded as 
\begin{align*}
(\bfB_\text{BB}\!+\!\Delta\bfB_\text{BB})^+(\bfB_\text{BG}\!+\!\Delta\bfB_\text{BG})&\!=\!(\Delta_e+\bfB_\text{BB}^+)(\bfB_\text{BG}\!+\!\Delta\bfB_\text{BG})\\
&=\bfB_\text{BB}^+\bfB_\text{BG}+\Delta_\text{error}, 
\end{align*}
where $\Delta_\text{error}$ equals the remaining terms in the matrix product. Plugging the identity above in \eqref{eq: frequency divider errors1} yields us 
\begin{align}\label{eq: linear model with errors}
(\bfB_\text{BB}^+\bfB_\text{BG}+\Delta_\text{error})\domega_G(t)=-(\domega_B(t)+\mathbf{e}(t)). 
\end{align}

In \eqref{eq: linear model with errors}, the unknown variables are $\domega_{\bfG}(t)$ and errors $\Delta_\text{error}$ and $\mathbf{e}(t)$. In general, $\mathbf{e}(t)$ is uncorrelated with $\Delta_\text{error}$ because the former is due to measurement uncertainty and the latter is the cumulative errors in the network parameters. 
As such, we estimate $\domega_{\bfG}(t)$ using simple TLS \cite{markovsky2007overview}.
This gives us
\begin{align}
    \widehat{\Delta}\mathbf{\omega}_\bfG(t)=-\mathbf{v}_{pq}v_{qq}^{-1}, 
\end{align}
where the vector $\mathbf{v}_{pq}$ consists of first $n$ elements of the $(n+1)$-th column of the matrix of right singular vectors, $\mathbf{V}$, of the augumented matrix $[\mathbf{\bfB_\text{BB}^+\bfB_\text{BG}}\,|\, -\domega_B(t)]$; and $v_{qq}$ is the $n+1$-th element of the $(n+1)$-th column of $\mathbf{V}$. 

{\color{black} In essence, we estimated rotor speed variations from erroneous bus frequencies and network parameters using the TLS technique. This method could also be used to infer rotor speeds (not just the oscillatory deviations), extending the framework in \cite{milano2018model}. Returning to our localization problem, the components of $\widehat{\Delta}\mathbf{\omega}_\bfG(t)$ will be close to zero except for the entries that correspond to the oscillatory sources.} 


\section{Simulation results}
This section presents simulation results on the IEEE 14-bus and 39-bus systems that showcase the utility of the TLS technique (see Section \ref{sec: multiple sources}) in localizing multiple sources. The topology and branch data (to compute $\bfB_\text{BUS}$), and the internal reactance (to compute $\bfB_\mathbf{BG}$ and $\bfB_\mathbf{GG}$) are given in \cite{milano2020frequency}. We compute the non-zero entries of $\bfB_\mathbf{GG}$ by taking the average of $q$- and $d$- axis reactances. For the test systems, we assume that each network bus has a PMU installed; and model the PMU data noise and the line parameter errors as zero mean Gaussian random variable\footnote{\color{black} Gaussian assumption is a worst case scenario because, as discussed earlier, the errors in practice are only bounded.} with variances $0.02$ and $0.3$, respectively. 

\smallskip 
\noindent \emph{\textbf{IEEE 14-bus}}: The five generator units labeled as $\{1,2,3,4,5\}$ are connected to the buses $\{1,2,3,6,8\}$, respectively. {\color{black} In this 
study, the rotor speed deviations are set manually\footnote{\color{black} A more realistic approach is to conduct a dynamic simulation study in which controllers/exciters 
inject forced oscillations into the power system, causing the machine rotor angles to oscillate. 
We defer such a detailed study to our future work; see also Section \ref{sec: conclusion}.}: that is, we model}  $\domega_{\bfG,1}(t)=\sin(1.41t)+\sin(0.2t)$ and $\domega_{\bfG,5}(t)=a(t)\sin(1.13t)$, where $a(t)$ is linear in $t$ (see Fig.~\ref{fig:figure_label3}), and $\domega_{\bfG,j}(t)=0$, for $j\in \{2,3,4\}$. 

Fig.~\ref{fig:subfigure41} shows the bus frequency measurements for select buses. From this plot alone, it is impossible to determine how many sources are present, let alone localizing them. Even if one were to guess the number of sources (based on oscillation frequencies), localizing all of them based on magnitudes is not possible. For instance, identifying the source at bus 8 (in green) is highly unlikely.
Fig.~\ref{fig:subfigure42} shows the TLS estimated rotor speed variations. The TLS method correctly identified the generators inducing the FOs and the waveform.

\smallskip 
\noindent \emph{\textbf{IEEE 39-bus}}: The ten generators labeled $\{1,2,\ldots,10\}$ are connected to the buses $\{30,31,\ldots,39\}$. The speed variations are same as in the case of IEEE 14-bus system. In this system, the sources are at buses $30$ and $34$. Fig.~\ref{fig:subfigure51} shows the bus frequency measurements for select buses. While one could reasonably argue that a source could be at bus $30$, locating the other source would be impossible. Fig.~\ref{fig:subfigure52} shows the estimated rotor speed variations by the TLS technique, which accurately identified the two sources.

\smallskip 
{\color{black} \noindent \emph{\textbf{Omissions and justification}}: One valid criticism against our proposed methods is that their performance is not compared with the existing approaches.
Regarding this, please note that our
first goal was to exploit the properties of \emph{frequency divider} formula to enable source localization, obviating the need for detailed machine models. The second was to account for the parameter and measurement errors in the divider formula. Simulations in this section are a testimony to both our goals, and comparisons with existing methods, while useful, is cursory. Hence, we did not include them.}

\begin{figure}
    \centering
    \includegraphics[width=1.0\linewidth]{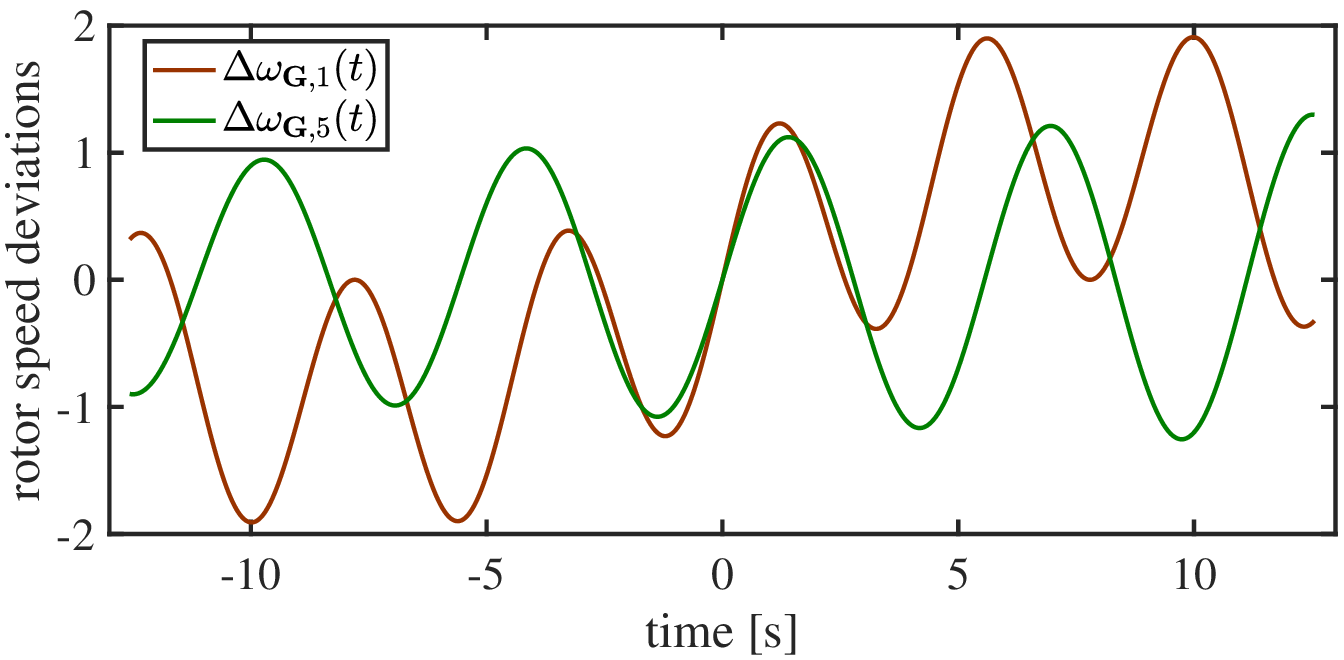}
    \caption{\small Oscillations manifested in the rotor speeds of generators.}
    \label{fig:figure_label3}
\end{figure}

\begin{figure}[!t]
  \centering
  \subfigure[Bus frequency variations at selected network buses.]{\includegraphics[width=1.0\linewidth]{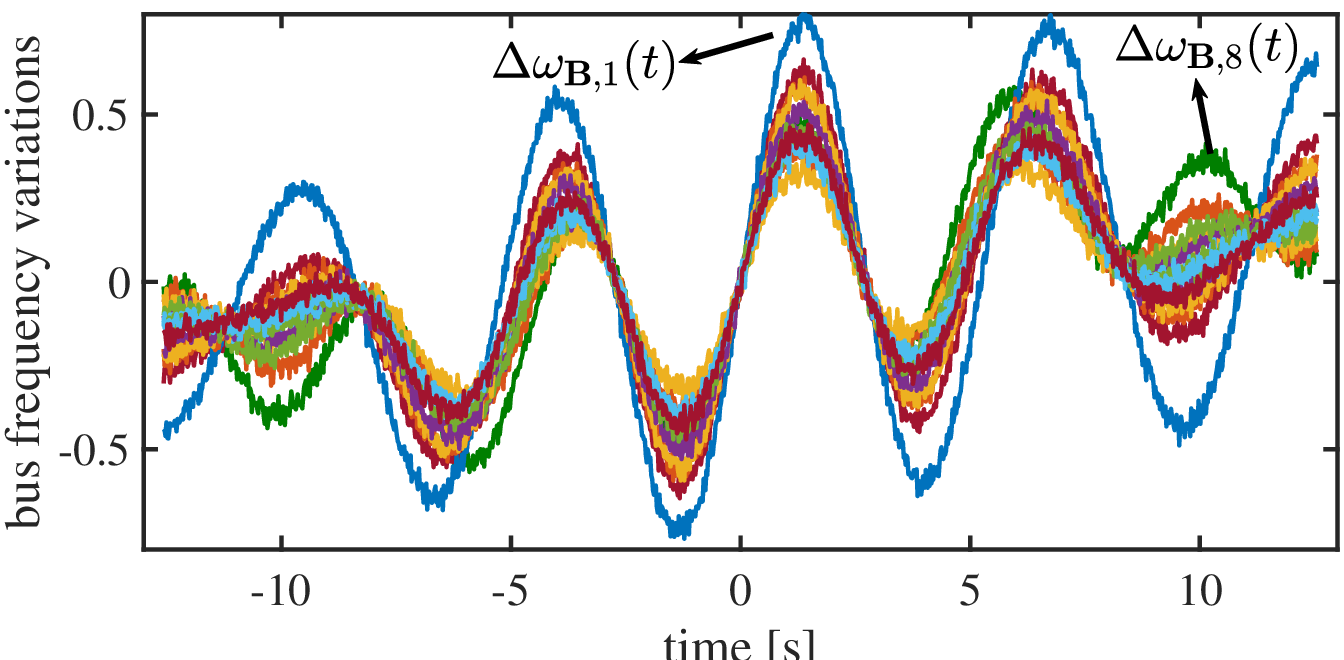}%
  \label{fig:subfigure41}}
  \hfil
  \subfigure[TLS estimates of rotor speed variations.]{\includegraphics[width=1.0\linewidth]{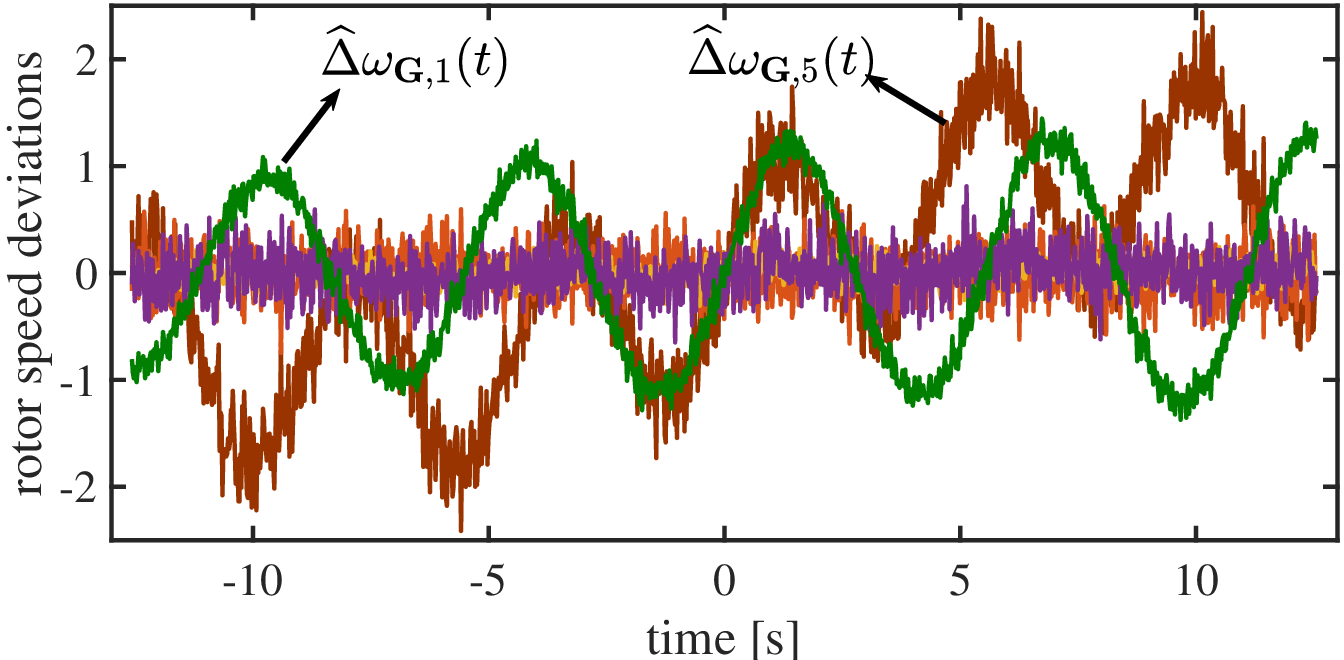}%
  \label{fig:subfigure42}}
  \caption{\small Localizing sources via TLS: IEEE-14 bus system.}
  \label{fig:figure_label4}
\end{figure}

\begin{figure}[!t]
  \centering
  \subfigure[Bus frequency variations at selected network buses.]{\includegraphics[width=1.0\linewidth]{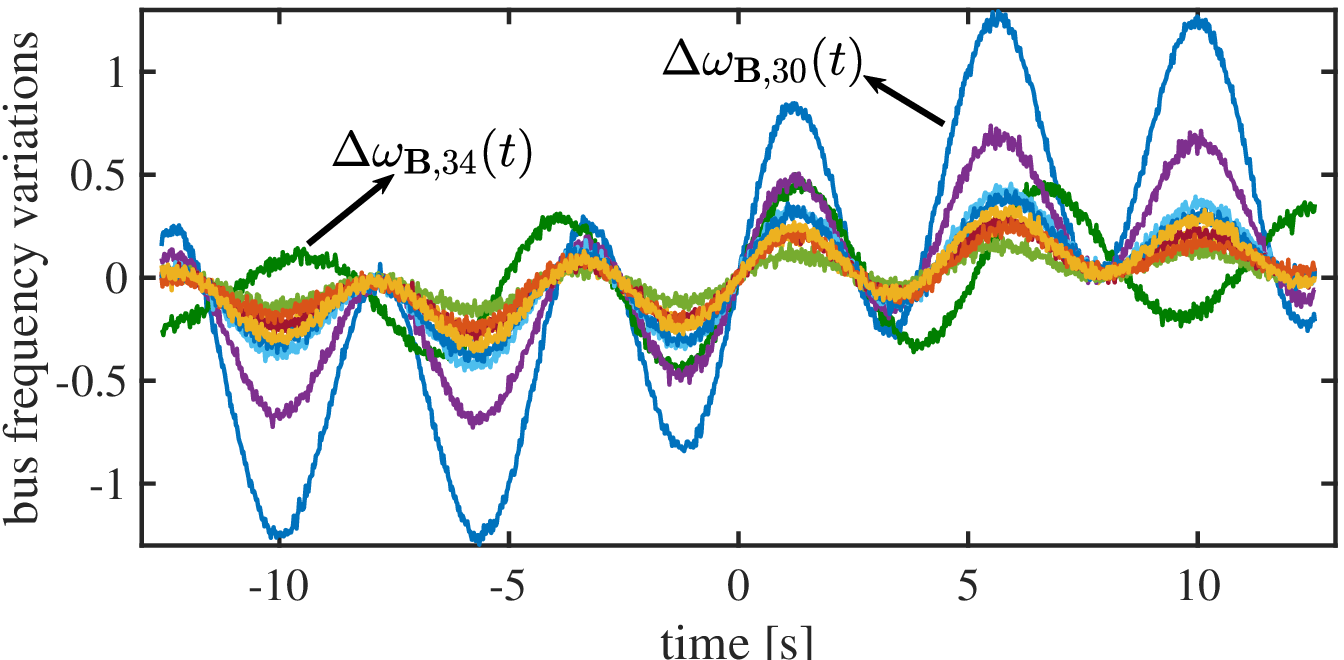}%
  \label{fig:subfigure51}}
  \hfil
  \subfigure[TLS estimates of rotor speed variations.]{\includegraphics[width=1.0\linewidth]{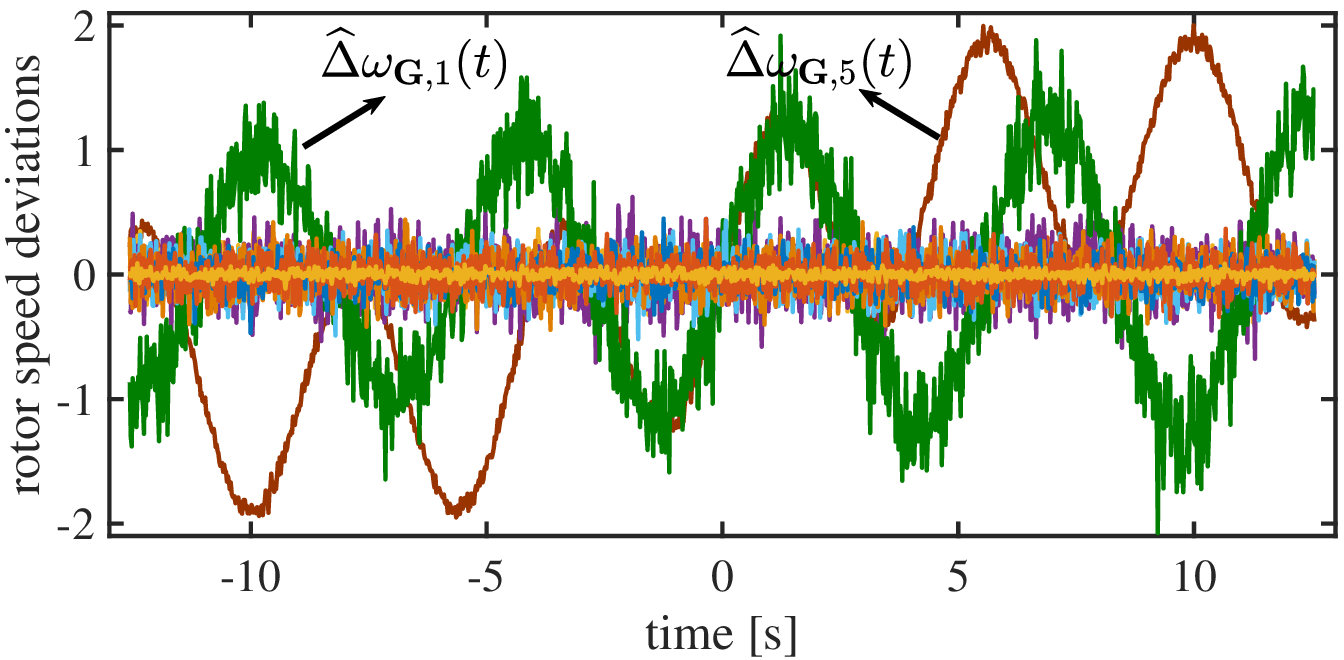}%
  \label{fig:subfigure52}}
  \caption{\small Localizing sources via TLS: IEEE-39 bus system.}
  \label{fig:figure_label5}
\end{figure}

\section{Conclusion}\label{sec: conclusion}
We presented two practical methods for localizing single and multiple oscillatory sources. Both methods are generic, and {\color{black} 
exploit 
certain
matrix-theoretic properties of the \emph{frequency divider} formula.} Thus, they could be applied to localize any unwanted disturbances, oscillatory or otherwise, and could be applied irrespective of the type of the FO source, synchronous-based or converter-based. The first method concerning single FO source relied on the magnitudes of bus frequencies of PMUs and was extremely fast.  The second method concerning multiple FO source locations relied on the TLS technique. 
The computational burden of the second method comes from computing the inverse of the admittance matrices and performing the SVD operation for the TLS. Both these operations could be sped up by utilizing the sparsity of the matrices.


{\color{black} Our future work will focus on three aspects: (i) develop strategies to localize sources without requiring PMUs to be placed on all the network buses; 
(ii) develop a localization method based on \eqref{eq: frequency divider errors}, thereby obviating the need for computing the pseudo-inverse; (iii) compare and contrast our methods with existing methods on larger systems, including the WECC 179-bus system and the renewable-rich WECC 240-bus system \cite{yuan2020developing}, and the real-world data sets provided in 
\cite{maslennikov2022creation}.}

	\bibliographystyle{IEEEtran}
	\bibliography{BIB.bib}

\end{document}